\documentclass[a4paper,10pt]{article}

\newcommand{\thetitle}{The fundamental theorem of asset pricing under proportional transaction costs}
\newcommand{\theauthor}{Alet Roux}
\newcommand{\theinstitution}{University of York}

\title{\thetitle}
\author{\theauthor\thanks{Department of Mathematics, \theinstitution. Email: ar521\@york.ac.uk.}}

\usepackage{amsfonts}

\usepackage{amsmath}
\DeclareMathOperator{\successors}{succ}
\DeclareMathOperator{\conv}{conv}
\DeclareMathOperator{\cone}{cone}

\usepackage{harvard}
\citationstyle{dcu}
\usepackage{amsthm}
\theoremstyle{definition}
\newtheorem{theorem}{Theorem}
\newtheorem{lemma}[theorem]{Lemma}
\newtheorem{proposition}[theorem]{Proposition}
\theoremstyle{definition}
\newtheorem{definition}[theorem]{Definition}
\theoremstyle{remark}
\newtheorem{remark}[theorem]{Remark}
\newtheoremstyle{citing}% name
  {3pt}%      Space above, empty = `usual value'
  {3pt}%      Space below
  {\itshape}% Body font
  {}%         Indent amount (empty = no indent, \parindent = para indent)
  {\bfseries}% Thm head font
  {.}%        Punctuation after thm head
  {.5em}%     Space after thm head: " " = normal interword space;
  {\thmnote{#3}}% Thm head spec
\theoremstyle{citing}
\newtheorem*{varthm}{}% all text supplied in the note

\begin{document}

\maketitle

\begin{abstract}
We extend the fundamental theorem of asset pricing to a model where the risky stock is subject to proportional transaction costs in the form of bid-ask spreads and the bank account has different interest rates for borrowing and lending. We show that such a model is free of arbitrage if and only if one can embed in it a friction-free model that is itself free of arbitrage, in the sense that there exists an artificial friction-free price for the stock between its bid and ask prices and an artificial interest rate between the borrowing and lending interest rates such that, if one discounts this stock price by this interest rate, then the resulting process is a martingale under some non-degenerate probability measure. Restricting ourselves to the simple case of a finite number of time steps and a finite number of possible outcomes for the stock price, the proof follows by combining classical arguments based on finite-dimensional separation theorems with duality results from linear optimisation.
\end{abstract}

The fundamental theorem of asset pricing characterises models of financial markets without arbitrage or free lunch, i.e.~the making of risk-free profit without initial investment. It is well known that a classical friction-free model containing a risky stock and a bank account admits no arbitrage if and only if there exists a probability measure on the model under which the stock price, discounted by the interest rate on the bank account, is a martingale. In turn, the collection of such probability measures play a fundamental role in the pricing of contingent claims and derivative securities.

The aim of this paper is to extend the fundamental theorem to a model where the risky stock is subject to proportional transaction costs in the form of bid-ask spreads, so that a share can always be sold for the bid price and bought for the ask price, and the bank account has different interest rates for borrowing and lending. We show that such a model is free of arbitrage if and only if one can embed in it a friction-free model that is itself free of arbitrage, in the sense that there exists an artificial friction-free price for the stock between its bid and ask prices and an artificial interest rate between the borrowing and lending interest rates such that, if one discounts this stock price by this interest rate, then the resulting process is a martingale under some non-degenerate probability measure.

Section~\ref{sec:fundamentals:the-model} is a formal introduction to a model of a financial market consisting of a risky asset (a stock) and a risk-free asset (a bond, referred to as cash). The model is subject to friction in the sense that proportional transaction costs are payable on the stock, and cash accrues interest at different rates of interest, depending on whether the cash balance is positive or negative. Participants in this financial market trade in both assets, subject to reasonable restrictions. We demonstrate in Section~\ref{sec:fundamentals:trading-arbitrage} that in order to prohibit arbitrage in the presence of transaction costs, it is sufficient to be able to embed in our model an artificial friction-free model, consisting of stocks and cash, that is again arbitrage-free. We refer to such an embedded friction-free model as an equivalent martingale triple (see Definition~\ref{def:fundamentals:trading-arbitrage:equivalent-martingale-triple}). In Section~\ref{sec:fundamentals:ftap} we show that, conversely, if our model admits no arbitrage, then it must contain an equivalent martingale triple. The main result of this paper is Theorem~\ref{th:fundamentals:ftap}, which extends the well-known fundamental theorem of asset pricing to the present case.

\section{The model}
\label{sec:fundamentals:the-model}

Consider a discrete-time market model with finite time horizon \(T\in\mathbb{N}\) and finite state space \(\Omega\), and let \(\mathbb{T}:=\{0,\ldots,T\}\) be the set of trading dates. Let~\(\mathbb{Q}\) be a probability measure on the collection \(2^{\Omega}\) of all subsets of \(\Omega\) such that~$\mathbb{Q}(\{\omega\})>0$ for all $\omega\in\Omega.$ We assume given a filtration \((\mathcal{F}_t)_{t\in\mathbb{T}}\) on the set~\(\Omega\), with \mbox{\(\mathcal{F}_0=\{\emptyset,\Omega\}\)} being the trivial \(\sigma\)-field and \(\mathcal{F}_T=2^{\Omega}\). For any \(t\in\mathbb{T}\), we denote by \(\Omega_t\) the collection of atoms of \(\mathcal{F}_t\). We shall often identify \(\mathcal{F}_t\)-measurable random variables with functions defined on \(\Omega_t\). Observe that \(\Omega_0=\{\Omega\}\) and \(\Omega_T\) is the collection of single-element subsets of~\(\Omega\).

It is natural to view such a model as an \emph{event tree}. We identify the \emph{nodes} of the tree at time \(t\in\mathbb{T}\) with the elements of \(\Omega_t\). By virtue of the properties of~\((\mathcal{F}_t)_{t\in\mathbb{T}}\), there is only one \emph{root node} at time \(0\), and it is associated with the certain event. Likewise, there is a one-to-one correspondence between \emph{terminal nodes} at time \(T\) and scenarios \(\omega\in\Omega\). At any time \(t\in\mathbb{T}\backslash\{T\}\) and any node \mbox{\(\nu\in\Omega_t\)}, we call a node \(\mu\in\Omega_{t+1}\) a \emph{successor} to \(\nu\) if \(\mu\subseteq\nu\). Denote by
\[\successors\nu:=\{\mu\in\Omega_{t+1}|\mu\subseteq\nu\}\]
the collection of successor nodes of \(\nu\).

The model consists of two securities, namely a risk-free \emph{bond} and a risky \emph{stock}. Proportional transaction costs apply to trading in the stock, i.e.\ at any time \(t\in\mathbb{T},\) the stock can be sold for the \emph{bid price} \(S^b_t\) and bought for the \emph{ask price} \(S^a_t\). We assume that \(S^b\) and \(S^a\) are adapted to \((\mathcal{F}_t)_{t\in\mathbb{T}}\), and that
\begin{equation}\label{eq:fundamentals:the-model:stock-no-arbitrage}
S^a_t\ge S^b_t>0\text{ for }t\in\mathbb{T}.
\end{equation}

Option pricing in models with proportional transaction costs have been studied by
%\citeasnoun{garman_ohlson1981}, \citeasnoun{gilster_lee1984}, 
\citeasnoun{leland1985}, \citeasnoun{hodges_neuberger1989}, \citeasnoun{davis_norman1990}, \citeasnoun{merton1990}, 
%\citeasnoun{dermody_rockafellar1991}, 
\citeasnoun{boyle_vorst1992}, \citeasnoun{bensaid_lesne_pages_scheinkman1992}, 
%\citeasnoun{avellaneda_paras1994a}, \citeasnoun{gennotte_jung1994}, 
\citeasnoun{hoggard_whalley_wilmott1994}, 
%\citeasnoun{davis_zariphopoulou1995}, 
%\citeasnoun{jouini_kallal1995a}, 
%\citeasnoun{kusuoka1995}, 
%\citeasnoun{naik1995}, 
\citeasnoun{soner_shreve_cvitanic1995},
\citeasnoun{cvitanic_karatzas1996}, 
%\citeasnoun{grannan_swindle1996}, 
%\citeasnoun{toft1996}, 
%\citeasnoun{clewlow_hodges1997}, 
%\citeasnoun{jouini_koehl_touzi1997}, 
%\citeasnoun{kabanov_safarian1997}, \citeasnoun{levental_skorohod1997}, 
%\citeasnoun{mercurio_vorst1997}, 
\citename{perrakis_lefoll1997} \citeyear{perrakis_lefoll1997,perrakis_lefoll2000,perrakis_lefoll2004}, \citename{stettner1997} \citeyear{stettner1997,stettner2000}, 
%\citeasnoun{tourin_zariphopoulou1997}, 
%\citeasnoun{whalley_wilmott1997},
%\citeasnoun{ahn_dayal_grannan_swindle1998}, 
%\citeasnoun{barles_soner1998}, 
%\citeasnoun{rutkowski1998}, 
\citename{constantinides_zariphopoulou1999} \citeyear{constantinides_zariphopoulou1999,constantinides_zariphopoulou2001}, 
%\citeasnoun{cvitanic_pham_touzi1999}, 
%\citeasnoun{keppo_peura1999}, 
%\citename{kocinski1999} \citeyear{kocinski1999,kocinski2001b,kocinski2004}, 
\citename{koehl_pham_touzi1999} \citeyear{koehl_pham_touzi1999,koehl_pham_touzi2001},  
%\citeasnoun{touzi1999}, 
%\citeasnoun{jouini2000}, 
%\citeasnoun{koehl_pham2000}, 
%\citeasnoun{loewenstein2000},
\citeasnoun{chalasani_jha2001}, %\citeasnoun{cvitanic_wang2001}, 
%\citeasnoun{grandits_schachinger2001}, 
\citeasnoun{ortu2001}, %\citeasnoun{palmer2001a},
%\citename{palmer2001a} \citeyear{palmer2001a,palmer2001b}, 
%\citeasnoun{reisman2001},
%\citeasnoun{bouchard2002}, 
%\citeasnoun{constantinides_perrakis2002}, 
%\citename{damgaard2003} \citeyear{damgaard2003,damgaard2006}, 
\citeasnoun{delbaen_kabanov_valkeila2002}, %\citeasnoun{liu_loewenstein2002}, \citeasnoun{zhang_xu_deng2002},
%\citeasnoun{baran2003}, 
\citeasnoun{gondzio_kouwenberg_vorst2003}, %\citename{monoyios2003} \citeyear{monoyios2003,monoyios2004}, 
%\citeasnoun{perrakis_czerwonko2003}, \citeasnoun{zhao_ziemba2003}
%\citeasnoun{best_hlouskova2005}, 
\citeasnoun{sass2005}, \citeasnoun{albanese_tompaidis2006}, \citeasnoun{chen_palmer_sheu2006a} and many others. In the vast majority of these papers the authors assume the existence of a classical friction-free stock price process \(S\), and introduce transaction costs by setting
\begin{align*}
S^a &:= (1+\lambda)S, &
S^b &:= (1-\mu)S
\end{align*}
for some constant coefficients \(\lambda\ge0\) and \(\mu\in[0,1)\). Here we follow the more general approach of \citeasnoun{jouini_kallal1995a} by modelling bid and ask prices as two different processes. As pointed out by \citeasnoun[pp.~548--549]{jouini2000}, the bid-ask spread on stock prices, though interpreted as variable transaction costs, can also be explained by the buying and selling of limit orders. The bid and ask prices of the stock are the prices for which a buyer or seller is guaranteed to find an immediate counter-party to execute the trade, i.e.~the bid and ask prices ensure liquidity in the market.

% \begin{remark}
% Another approach that has recently been developed by \citeasnoun{kabanov1999}, \citeasnoun{kabanov_stricker2001b}, \citeasnoun{schachermayer2004} and others is especially applicable to models of currency markets. Prices of commodities are modelled as a matrix process \(\Pi\), where \(\pi_{ij}\) is the amount in currency \(i\) that must be sacrificed in order to obtain one unit of currency \(j\). The assumption
% \[\pi_{ij} \pi_{ji} > 1\]
% is equivalent to the presence of proportional transaction costs on transfers between currency \(i\) and \(j\). Results in this framework generally do not depend on the existence of a num\'eraire; refer to the work by \citeasnoun{bouchard_touzi2000}, \citeasnoun{bouchard_kabanov_touzi2001}, \citeasnoun{kabanov_last2002}, \citeasnoun{kabanov_rasonyi_stricker2002}, \citeasnoun{bentahar_bouchard2005}, \citeasnoun{bouchard_temam2005} and others.
% \end{remark}

Our model also contains a \emph{bank account} with different lending and borrowing rates, described by predictable interest rate processes $r^c$ (for borrowing) and $r^d$ (for lending). For convenience let $r^d_0 := 0$ and $r^c_0 := 0$. We assume that $r^c$ and~$r^d$ satisfy
\begin{equation}
r^c_t \ge r^d_t > -1\text{ for } t\in\mathbb{T}. \label{eq:fundamentals:the-model:bond-no-arbitrage}
\end{equation}
Define
\[\varrho_t(\alpha):=\alpha^+\left(1+r^d_t\right)-\alpha^-\left(1+r^c_t\right)\]
for \(t\in\mathbb{T}\) and \(\alpha\in\mathbb{R}\), where
\(x^+\equiv\max\{x,0\}\) and \(x^-\equiv\max\{-x,0\}\). An investment of \(\alpha\) units of cash in the bank account at time \((t-1)\in\mathbb{T}\backslash\{T\}\) accumulates to \(\varrho_t(\alpha)\) in cash at time \(t\).

\begin{remark}\label{rem:fundamentals:the-model:constrains-costs}
Over single periods, the presence of different borrowing and lending rates can be viewed as proportional transaction costs on bonds. Indeed, at time \(t\in\mathbb{T}\backslash\{T\}\), a bond with face value \(1\) expiring at time \(t+1\) may be purchased at \(1/(1+r^d_{t+1})\) and sold at \(1/(1+r^c_{t+1})\le1/(1+r^d_{t+1})\).
\end{remark}

For convenience, we write
\(R^d\equiv1+r^d\) and \(R^c\equiv1+r^c\), and define the associated credit and deposit \emph{deflator processes} by
\begin{align*}
B^d_t &:=\prod_{s=0}^t \left(1 + r^d_s\right)=\prod_{s=0}^t R^d_s, &
B^c_t &:=\prod_{s=0}^t \left(1 + r^c_s\right)=\prod_{s=0}^t R^c_s
\end{align*}
for $t\in\mathbb{T}$.
Since there is a one-to-one correspondence between descriptions of the bank account in terms of interest rates and deflator process, we adopt the following convention. We reserve \(B\) for the deflator process associated with a positive predictable process \(R\), i.e.\
\[B_t \equiv\prod_{s=0}^t R_s\text{ for }t\in\mathbb{T}.\]
We use these definitions interchangeably, with decorations to indicate the correspondence between individual processes named \(R\) and \(B\).

Different borrowing and lending rates are often treated within the realm of constraints on portfolio holdings. Indeed, it is natural to treat cash holdings as two separate cash accounts, one of which may only have a non-negative balance and accrues interest at the lending rate, and another that may only have a non-positive balance and accrues interest at the borrowing rate. In this regard, consider the contributions of \citeasnoun{fleming_zariphopoulou1991}, \citename{cvitanic_karatzas1992} \citeyear{cvitanic_karatzas1992,cvitanic_karatzas1993}, \citeasnoun{edirisinghe_naik_uppal1993}, \citeasnoun{he_pages1993}, \citeasnoun{browne1995}, \citeasnoun{jouini_kallal1995b}, \citename{karatzas_kou1996} \citeyear{karatzas_kou1996,karatzas_kou1998}, \citeasnoun{cuoco1997}, \citeasnoun{carassus_jouini1998}, \citeasnoun{elkaroui_jeanblancpicque1998}, \citeasnoun[Chapter~6]{karatzas_shreve1998}, \citeasnoun{pham_touzi1999}, \citeasnoun{bensoussan_julien2000}, \citeasnoun{tepla2000}, \citeasnoun{bizid_jouini2001}, \citeasnoun{cvitanic2001}, \citeasnoun{subramanian2001}, \citeasnoun{soner_touzi2002}, \citename{napp2001} \citeyear{napp2001,napp2003}, \citeasnoun{zhang_wang_deng2004}, \citeasnoun{detemple_rindisbacher2005}, \citeasnoun{melnikov_petrachenko2005}, \citeasnoun{rokhlin2005} and others.

%Finally, we call our model \emph{binary} if every non-terminal node \(\nu\) has exactly two successors, which we denote by \(\nu\mathrm{u}\) and \(\nu\mathrm{d}\). Thus for all \(t\in\mathbb{T}\) the collection \(\Omega_t\) of nodes at time \(t\) in a binary model may be represented as the collection of all sequences of length \(t\) comprised of the symbols \(\mathrm{u}\) and \(\mathrm{d}\). For \(t\in\mathbb{T}\), we say that an \(\mathcal{F}_t\)-measurable random variable \(Z\) in a binary tree model is \emph{path-independent} at time \(t\) if its value \(Z(\nu)\) at every node \(\nu\in\Omega_t\) is independent of the order in which the symbols \(\mathrm{u}\) and \(\mathrm{d}\) appear in the sequential representation of \(\nu\). An adapted (resp.~predictable) process \(Z\) is called path-independent if its value \(Z_t\) (resp.~\(Z_{t+1}\)) at every time \(t\in\mathbb{T}\) is path-independent at time \(t\). We call a binary tree model \emph{recombinant} if the processes \(S^a\), \(S^b\), \(R^d\) and \(R^c\) are all path-independent.

\section{Trading and arbitrage}
\label{sec:fundamentals:trading-arbitrage}

Market agents trade in cash and stock. A market agent is allowed to change into any \emph{portfolio} \((\xi,\zeta)\) consisting of \(\xi\) in cash and \(\zeta\) units of stock at any trading date \(t\in\mathbb{T}\): the only (realistic) restrictions are that agents are limited to trades that they can afford, and are not prescient. A market agent seeking to avoid or minimise transaction costs will prefer to aggregate stock trades in order to avoid the simultaneous purchase and sale of stock. Likewise, if the credit interest rate exceeds the deposit rate, then it is not optimal, from the point of view of wealth maximisation, to borrow cash at the same time as having a cash deposit. Motivated by this, and for simplicity of notation, we do not allow either the simultaneous buying and selling of stock or simultaneous long and short positions in cash.

The \emph{liquidation value} \(\vartheta_t\) at time \(t\in\mathbb{T}\) of a portfolio \((\xi,\zeta)\) is defined as
\[\vartheta_t(\xi,\zeta):=\xi + \zeta^+S^b_t - \zeta^-S^a_t.\]
The cost of establishing the portfolio \((\xi,\zeta)\) at time \(t\) from a zero endowment in stock and cash is then
\[\xi + \zeta^+S^a_t - \zeta^-S^b_t= -\left[-\xi + \zeta^-S^b_t - \zeta^+S^a_t \right] =-\vartheta_t(-\xi,-\zeta).\]

%\begin{definition}[Trading strategy]
A \emph{trading strategy} $(\alpha,\beta)$ is an $\mathbb{R}^2$-valued predictable random process with $\alpha$ denoting the cash holding and $\beta$ denoting the holding in the risky asset.
%\end{definition}
We assume that $\alpha_0$ is the initial cost of setting up the trading strategy $(\alpha,\beta)$, and we set $\beta_0 := 0$ to indicate that the initial endowment of an investor does not contain any stock. It will sometimes be convenient to assign to each trading strategy $(\alpha,\beta)$ a final portfolio $(\alpha_{T+1},\beta_{T+1})$, into which an investor changes at time \(T\). Unless specified otherwise, we assume that the investor liquidates his/her position in stock at time \(T\), i.e.\ \(\beta_{T+1}=0\) and
\begin{equation}\label{eq:fundamentals:trading-arbitrage:cash-liquidation-alpha}
\alpha_{T+1}=\vartheta_T\left(\varrho_T(\alpha_T),\beta_T\right),
\end{equation}
so that
\[\vartheta_T(\alpha_{T+1},\beta_{T+1})=\vartheta_T\left(\varrho_T(\alpha_T),\beta_T\right).\]

Our interest is in strategies that require no additional injection of funds after time~\(0\). Such strategies may be characterised in terms of the cash flows that they generate. 

\begin{definition}[Self-financing strategy]
A \emph{self-financing trading strategy} is a trading strategy \((\alpha,\beta)\) such that \(\beta_0=0\) and
\begin{equation} \label{eq:fundamentals:trading-arbitrage:self-financing}
\vartheta_t(\varrho_t(\alpha_t)-\alpha_{t+1},\beta_t-\beta_{t+1})\ge0\text{ for }t\in\mathbb{T}.
\end{equation}
%If the inequality in~\eqref{eq:fundamentals:trading-arbitrage:self-financing} is replaced by an equality, then the strategy \((\alpha,\beta)\) is called \emph{strictly self-financing}.
\end{definition}

\begin{remark}\label{rem:fundamentals:trading-arbitrage:wedging}
The self-financing condition~\eqref{eq:fundamentals:trading-arbitrage:self-financing} is equivalent to having
\[\alpha_tR + \beta_tx\ge \alpha_{t+1} + \beta_{t+1}x\]
for all \(x\in\left\{S^b_t,S^a_t\right\}\) and \(R\in\left\{R^d_t,R^c_t\right\}\), equivalently, for all \(x\in\left[S^b_t,S^a_t\right]\) and \(R\in\left[R^d_t,R^c_t\right]\).
\end{remark}

Denote the collection of self-financing trading strategies by \(\Phi\). We have the following result.

\begin{proposition} \label{prop:fundamentals:trading-arbitrage:Phi-convex-polyhedral}
The collection \(\Phi\) of self-financing strategies is a convex polyhedral cone in \(\mathbb{R}^{2m+1}\), where \(m\) is the number of nodes in the model.
\end{proposition}

\begin{remark}
The dimension \(2m+1\) in Proposition~\ref{prop:fundamentals:trading-arbitrage:Phi-convex-polyhedral} may be reduced upon making further assumptions on the nature of trading strategies. To see this, fix any strategy \((\alpha,\beta)\in\Phi\). On the one hand, if one assumes that there is no consumption at time \(0\), i.e.
\[\alpha_0=\alpha_1+\beta^+_1S^a_0-\beta^-_1S^b_0,\]
then we may identify \((\alpha,\beta)\) with a vector in \(\mathbb{R}^{2m}\). On the other hand, if one forces stock liquidation without consumption at time \(T\), i.e.\ \(\beta_{T+1}=0\) and \(\alpha_{T+1}\) is given by~\eqref{eq:fundamentals:trading-arbitrage:cash-liquidation-alpha}, then we may identify \((\alpha,\beta)\) with a vector in \(\mathbb{R}^{2n+1}\), where \(n\) is the number of non-terminal nodes in the model.
\end{remark}

The model allows \emph{arbitrage} if it allows an investor to achieve risk-free profit without any initial investment. An arbitrage opportunity is defined as follows.

\begin{definition}[Arbitrage]
An \emph{arbitrage opportunity} is a self-financing strategy \((\alpha,\beta)\in\Phi\) such that
\begin{align*}
\alpha_0&\le0,& \vartheta_T(\alpha_{T+1}, \beta_{T+1})&\ge0, & \mathbb{Q}(\vartheta_T(\alpha_{T+1}, \beta_{T+1})>0)&>0.
\end{align*}
\end{definition}

\begin{remark}
If one allows market agents to buy and sell stocks simultaneously and to have cash loans and deposits at the same time, then it is not strictly necessary to assume explicitly that the ask price exceeds the bid price and the credit interest rate exceeds the deposit rate, as we have done in~\eqref{eq:fundamentals:the-model:stock-no-arbitrage} and~\eqref{eq:fundamentals:the-model:bond-no-arbitrage}. Indeed, these two inequalities would follow directly if one assumes that the model is free of arbitrage.
\end{remark}

In this section, we provide a sufficient condition for the absence of arbitrage. This condition is the existence of an equivalent martingale triple, which we define as follows.

\begin{definition}[Martingale triple]\label{def:fundamentals:trading-arbitrage:equivalent-martingale-triple}
A \emph{martingale triple} $(\mathbb{P},S,R)$ (equivalently, $(\mathbb{P},S,B)$) consists of a probability measure $\mathbb{P}$ on \(\Omega\), an adapted stock price process $S$ and predictable interest process $R$ (equivalently, deflator process \(B\)) satisfying
\begin{align*}
S^b &\le S \le S^a, &%\label{eq:(S,r)-wedge-S}\\
R^d &\le R \le R^c %\label{eq:(S,r)-wedge-R}
\end{align*}
such that $\frac{S}{B}$ is a martingale with respect to $\mathbb{P}$. If in addition the probability measure \(\mathbb{Q}\) is equivalent to \(\mathbb{P}\), then $(\mathbb{P},S,R) \equiv (\mathbb{P},S,B)$ is called an \emph{equivalent martingale triple}.
\end{definition}

Recall that a classical friction-free model with stock price process \(S\) and deflator process \(B\) is free of arbitrage if and only if there exists a probability measure \(\mathbb{P}\) equivalent to \(\mathbb{Q}\) with respect to which \(\frac{S}{B}\) is a martingale \citeaffixed[Theorem~6.1.1]{delbaen_schachermayer2006}{e.g.}. Therefore a martingale triple \((\mathbb{P},S,R)\) can be regarded as an artificial friction-free, arbitrage-free model that is embedded in the model of Section~\ref{sec:fundamentals:the-model} in the sense that \(S\) is between the bid and ask prices of the stock, and \(R\) is between the credit and deposit interest processes.

Let $\overline{\mathcal{P}}$ be the collection of martingale triples, and \(\mathcal{P}\subseteq\overline{\mathcal{P}}\) the collection of equivalent martingale triples.

\begin{lemma} \label{lem:fundamentals:trading-arbitrage:super-martingale-property}
Fix $(\mathbb{P},S,R)\equiv(\mathbb{P},S,B) \in \overline{\mathcal{P}}$, as well as a self-financing strategy $(\alpha,\beta)\in\Phi$. The process \(\left(\frac{1}{B_t}[\alpha_t R_t + \beta_t S_t]\right)_{t\in\mathbb{T}}\) is a super-martingale with respect to $\mathbb{P}$.
\end{lemma}

\begin{proof}
Fix \(t\in\mathbb{T}\backslash\{T\}\). From Remark~\ref{rem:fundamentals:trading-arbitrage:wedging} and the martingale property of \(\frac{S}{B}\), we obtain
\begin{multline*}
\frac{1}{B_t}[\alpha_t R_t + \beta_t S_t]
\ge\frac{1}{B_t}\left[\alpha_{t+1} + \beta_{t+1} S_t\right] =\frac{1}{B_t}\alpha_{t+1}+\beta_{t+1}\mathbb{E}_{\mathbb{P}}\left(\left.\frac{S_{t+1}}{B_{t+1}}\right|\mathcal{F}_t\right)\\
=\mathbb{E}_{\mathbb{P}}\left(\left.\frac{1}{B_{t+1}}[\alpha_{t+1} R_{t+1} + \beta_{t+1} S_{t+1}]\right|\mathcal{F}_t\right).\qedhere
\end{multline*}
\end{proof}

We conclude with the main result of this section.

\begin{proposition} \label{prop:fundamentals:trading-arbitrage:existence-implies-viability}
The existence of an equivalent martingale triple \((\mathbb{P},S,B)\) in \(\mathcal{P}\) implies lack of arbitrage.
\end{proposition}

\begin{proof}
Suppose, by contradiction, that there exists a self-financing strategy $(\alpha,\beta)\in\Phi$ with initial cost $\alpha_0 \le 0$ satisfying
\[\vartheta_T(\alpha_{T+1}, \beta_{T+1})\ge0\] and 
\begin{equation}\label{eq:fundamentals:trading-arbitrage:existence-implies-viability}
\mathbb{Q}(\vartheta_T(\alpha_{T+1}, \beta_{T+1})>0)>0.
\end{equation}
As $\mathbb{P}$ is equivalent to $\mathbb{Q}$, inequality~\eqref{eq:fundamentals:trading-arbitrage:existence-implies-viability} and Remark~\ref{rem:fundamentals:trading-arbitrage:wedging} imply that
\begin{multline*}
0 
< \mathbb{E}_{\mathbb{P}}\left(\frac{1}{B_T}\vartheta_T(\alpha_{T+1}, \beta_{T+1})\right)
=\mathbb{E}_{\mathbb{P}}\left(\frac{1}{B_T}\left[ \alpha_{T+1} + \beta_{T+1}^+S^b_T - \beta_{T+1}^-S^a_T \right]\right)\\
\le\mathbb{E}_{\mathbb{P}}\left(\frac{1}{B_T}\left[ \alpha_{T+1}  + \beta_{T+1}S_T\right]\right)\le\mathbb{E}_{\mathbb{P}}\left(\frac{1}{B_T}[\alpha_T R_T + \beta_T S_T]\right) \le\alpha_0 \le 0,
\end{multline*}
due to Lemma~\ref{lem:fundamentals:trading-arbitrage:super-martingale-property}.
\end{proof}

\section{Fundamental theorem}
\label{sec:fundamentals:ftap}

Loosely speaking, the fundamental theorem of asset pricing states that a
model of a financial market is free of arbitrage if and only if all the assets
in the model can be priced in a consistent way. The fundamental theorem for friction-free models was established by \citeasnoun{harrison_pliska1981} (finite state space) and \citeasnoun{dalang_morton_willinger1990} (infinite state space) for discrete time models (with several alternative proofs appearing subsequently), and for continuous-time models by \citename{delbaen_schachermayer1994} \citeyear{delbaen_schachermayer1994,delbaen_schachermayer1998}.

\citeasnoun[Theorem~3.2]{jouini_kallal1995a} were the first to extend the fundamental theorem to models with proportional transaction costs. With the current definition of arbitrage, \citeasnoun[Theorem 2 and Section 5]{kabanov_stricker2001b} and \citeasnoun[Theorem~3]{ortu2001} established an analogue of the result by \citeasnoun{harrison_pliska1981}, and \citeasnoun[Theorem~3.1]{zhang_xu_deng2002}, \citeasnoun[Theorem~1]{kabanov_rasonyi_stricker2002} and \citeasnoun[Theorem~1.7]{schachermayer2004} did the same for the work of \citeasnoun{dalang_morton_willinger1990}. In the current setting, these results show that a model with proportional transaction costs on the stock (and zero interest on cash) is free of arbitrage if and only if there exists an artificial stock price process \(S\) taking its values between the bid and ask prices of the stock, and a probability measure \(\mathbb{P}\) with bounded density such that \(S\) is a martingale under~\(\mathbb{P}\).

The fundamental theorem of asset pricing for markets with short sale constraints was first established by \citeasnoun[Theorem~2.1]{jouini_kallal1995b}. In their setting, holdings are required to be non-negative in some assets, and non-positive in others, and it therefore includes the case of different borrowing and lending rates. Most of the subsequent progress in establishing the fundamental theorem in models with investment constraints has been in discrete time, especially on extending the work of \citeasnoun{dalang_morton_willinger1990}. \citeasnoun[Theorem~2.4]{schurger1996} considered a model where all asset holdings are required to be non-negative. \citeasnoun[Theorem~4.2]{pham_touzi1999} and \citeasnoun[Lemma~3.1]{napp2003} also considered the case where the amounts invested in assets are constrained to lie in a closed convex cone. \citeasnoun[Theorem~3.1]{carassus_pham_touzi2001} required asset holdings to be in one closed convex set if the wealth of an investor is positive, and in another if negative. \citeasnoun[Theorem~1.1]{evstigneev_schurger_taksar2004} and \citeasnoun[Theorem~2]{rokhlin2005} established the fundamental theorem under random cone constraints. The setting of convex cone constraints on asset holdings contains the case of different borrowing and lending rates. Indeed, \citeasnoun[Corollary~4.1]{napp2003} showed that a model with different borrowing and lending rates (but no proportional transaction costs on the stock \(S\)) is free of arbitrage if and only if there exists an artificial interest rate process \(r\) (and associated deflator \(B\)) taking its values between the deposit and credit rates and a probability measure \(\mathbb{P}\) with bounded density such that \(\frac{S}{B}\) is a martingale under~\(\mathbb{P}\). 

The setting of convex cone constraints on asset holdings does extend to proportional transaction costs, but only in single-period models (\citename{pham_touzi1999}, \citeyear*{pham_touzi1999}, Application~3.1; \citename{napp2003}, \citeyear*{napp2003}, Corollary~4.2) (see also Remark~\ref{rem:fundamentals:the-model:constrains-costs}). \citeasnoun[Theorem~7.2]{dempster_evstigneev_taksar2006}, motivated by the theory of von Neumann-Gale systems, were the first to suggest a remedy to this shortcoming, thus making some headway towards a unified theory of asset pricing for models with constraints and transaction costs. In addition to placing constraints on asset holdings, they also require asset holdings in adjacent periods to jointly satisfy convex cone constraints. Consequently, they obtained the fundamental theorem for a discrete-time discrete-space model with proportional transaction costs (Theorem~9.3) (but no constraints on borrowing and lending). In our proofs below, we adapt their arguments to explicitly take account of different rates for borrowing and lending.

The main result of this paper admits the following simple formulation.

\begin{theorem} \label{th:fundamentals:ftap}
There is no arbitrage in the model if and only if it admits an equivalent martingale triple.
\end{theorem}

We showed in Proposition~\ref{prop:fundamentals:trading-arbitrage:existence-implies-viability} that the existence of an equivalent martingale triple implies absence of arbitrage. We now prove the converse implication in several steps, by combining classical arguments based on finite-dimensional separation theorems with duality results from linear optimisation. We first show in Lemma~\ref{lem:fundamentals:ftap:discount-factors} that lack of arbitrage implies the existence of a \emph{consistent discount factor}, which may be interpreted as the vector of shadow prices at time~\(T\). We then demonstrate in Lemma~\ref{lem:fundamentals:ftap:pricing-system} how consistent discount factors are used to construct a \emph{consistent pricing system}, i.e.\ a sequence of such shadow prices, one for each trading date. Lemma~\ref{lem:fundamentals:ftap:equivalent-pricing} concludes the proof by showing that the model admits an equivalent martingale triple if and only if it admits a consistent pricing system.

We first have a very familiar result. 

\begin{lemma}\label{lem:fundamentals:ftap:discount-factors}
If the model admits no arbitrage, then there exists an \((0,\infty)^2\)-valued \(\mathcal{F}_T\)-measurable random variable \(Z_T\equiv(Z^B_T,Z^S_T)\) such that
\[\mathbb{E}_{\mathbb{Q}}\left(Z_T\cdot \left(\alpha_{T+1},\beta_{T+1}\right)\right)-\alpha_0\le0\]
for all \((\alpha,\beta)\in\Phi\).
\end{lemma}

We call the random variable \(Z_T\) in Lemma~\ref{lem:fundamentals:ftap:discount-factors} a \emph{consistent discount factor}. We now continue our construction of sequences of shadow prices at the different trading dates of the model. 
%As we shall see in the proof of this result, the notion of consistent pricing systems is in some sense dual to that of self-financing trading strategies \citeaffixed[Remark~3.2]{cvitanic_karatzas1996}{see, for example,}.

\begin{lemma} \label{lem:fundamentals:ftap:pricing-system}
Suppose that the model admits a consistent discount factor. There exists an \((0,\infty)^2\)-valued adapted process \(Z\equiv(Z^B,Z^S)\) such that
\begin{equation}\label{eq:fundamentals:ftap:pricing-system:wedge-1}
S^b_t\le\frac{Z^S_t}{Z^B_t}\le S^a_t\text{ for }t\in\mathbb{T}
\end{equation}
and \(Z^S\) is a martingale under \(\mathbb{Q}\), and a predictable process \(R\) satisfying
\begin{equation}\label{eq:fundamentals:ftap:pricing-system:wedge-2}
0<R^d_t\le R_t\le R^c_t \text{ for } t\in\mathbb{T}
\end{equation}
such that \(BZ^R\) is a martingale under \(\mathbb{Q}\), where $B$ is the deflator process associated with $R$.
\end{lemma}

We refer to processes such as \(Z\) in Lemma~\ref{lem:fundamentals:ftap:pricing-system} as \emph{consistent pricing $($or valuation$)$ systems}. Our final result, which completes the proof of Theorem~\ref{th:fundamentals:ftap}, extends the arguments of \citeasnoun[Proposition~2.6]{harrison_pliska1981} and \citeasnoun[pp.~24--25]{schachermayer2004} from the friction-free case. Observe that the proof of Lemma~\ref{lem:fundamentals:ftap:equivalent-pricing} may readily be extended to show that there exists a one-to-one correspondence between equivalent martingale triples and consistent pricing systems.

\begin{lemma} \label{lem:fundamentals:ftap:equivalent-pricing}
The model admits an equivalent martingale triple if and only if it admits a consistent pricing system.
\end{lemma}

\begin{proof}
On the one hand, suppose that \(Z\equiv\left(Z^R,Z^S\right)\) is a consistent pricing system with associated predictable process \(R\) (and deflator \(B\)). Define now a probability measure \(\mathbb{P}\) equivalent to \(\mathbb{Q}\) by 
\[\mathbb{P}(\omega):=B_T(\omega)Z^R_T(\omega)\mathbb{Q}(\omega) \text{ for }\omega\in\Omega.\]
Also set
\[S_t:=\frac{Z^S_t}{Z^R_t} \text{ for }t\in\mathbb{T}.\]
Recalling that \(BZ^R\) and \(Z^S\) are martingales under \(\mathbb{Q}\), the tower property of conditional expectation, the predictability of \(B\) and our construction give
\begin{align*}
\mathbb{E}_{\mathbb{P}}\left(\left.S_{t+1}\right|\mathcal{F}_t\right)
&=\frac{\mathbb{E}_{\mathbb{Q}}\left(\left.S_{t+1}\frac{d\mathbb{P}}{d\mathbb{Q}}\right|\mathcal{F}_t\right)}{\mathbb{E}_{\mathbb{Q}}\left(\left.\frac{d\mathbb{P}}{d\mathbb{Q}}\right|\mathcal{F}_t\right)}\\
&=\frac{\mathbb{E}_{\mathbb{Q}}\left(\left.S_{t+1}B_TZ^R_T\right|\mathcal{F}_t\right)}{\mathbb{E}_{\mathbb{Q}}\left(\left.B_TZ^R_T\right|\mathcal{F}_t\right)}\\
&=\frac{\mathbb{E}_{\mathbb{Q}}\left(\left.S_{t+1}B_{t+1}Z^R_{t+1}\right|\mathcal{F}_t\right)}{B_tZ^R_t}\\
&=R_{t+1}\frac{\mathbb{E}_{\mathbb{Q}}\left(\left.Z^S_{t+1}\right|\mathcal{F}_t\right)}{Z^R_t}=R_{t+1}\frac{Z^S_t}{Z^R_t}=R_{t+1}S_t
\end{align*}
for \(t\in\mathbb{T}\backslash\{T\}\). The fact that \((\mathbb{P},S,R)\equiv(\mathbb{P},S,B)\in\mathcal{P}\) then follows from inequalities~\eqref{eq:fundamentals:ftap:pricing-system:wedge-1}--\eqref{eq:fundamentals:ftap:pricing-system:wedge-2}.

On the other hand, suppose that \((\mathbb{P},S,R)\equiv(\mathbb{P},S,B)\in\mathcal{P}\). Let
\begin{align*}
Z^R_T&:=\frac{1}{B_T}\frac{d\mathbb{P}}{d\mathbb{Q}}, &
Z^R_t&:=R_{t+1}\mathbb{E}_{\mathbb{Q}}\left(\left.Z^R_{t+1}\right|\mathcal{F}_t\right)
\end{align*}
for \(t\in\mathbb{T}\backslash\{T\}\) and
\[Z^S_t:=S_tZ^R_t \text{ for } t\in\mathbb{T}.\]
Inequality~\eqref{eq:fundamentals:ftap:pricing-system:wedge-1} and the martingale property of \(BZ^R\) under \(\mathbb{Q}\) clearly hold true; it remains for us to show that \(Z^S\) is a martingale with respect to \(\mathbb{Q}\). To this end, for \(t\in\mathbb{T}\backslash\{T\}\) we have
\begin{align*}
\mathbb{E}_{\mathbb{Q}}\left(\left.Z^S_{t+1}\right|\mathcal{F}_t\right)
&=\mathbb{E}_{\mathbb{Q}}\left(\left.S_{t+1}Z^R_{t+1}\right|\mathcal{F}_t\right)\\
&=\mathbb{E}_{\mathbb{Q}}\left(\left.\frac{S_{t+1}}{B_{t+1}}B_T Z^R_T\right|\mathcal{F}_t\right)\\
&=\mathbb{E}_{\mathbb{Q}}\left(\left.\frac{S_{t+1}}{B_{t+1}}\frac{d\mathbb{P}}{d\mathbb{Q}}\right|\mathcal{F}_t\right)\\
&=\mathbb{E}_{\mathbb{P}}\left(\left.\frac{S_{t+1}}{B_{t+1}}\right|\mathcal{F}_t\right)\mathbb{E}_{\mathbb{Q}}\left(\left.\frac{d\mathbb{P}}{d\mathbb{Q}}\right|\mathcal{F}_t\right)\\
&=\frac{S_t}{B_t}\mathbb{E}_{\mathbb{Q}}\left(\left.B_T Z^R_T\right|\mathcal{F}_t\right) =S_tZ^R_t=Z^S_t.\qedhere
\end{align*}
\end{proof}

\appendix
\section{Proofs of auxiliary results}

\begin{varthm}[Proposition \ref{prop:fundamentals:trading-arbitrage:Phi-convex-polyhedral}]
The collection \(\Phi\) of self-financing strategies is a convex polyhedral cone in \(\mathbb{R}^{2m+1}\), where \(m\) is the number of nodes in the model.
\end{varthm}

\begin{proof}
Since any strategy \((\alpha,\beta)\) is predictable with respect to the filtration \((\mathcal{F}_t)_{t\in\mathbb{T}}\), it follows that it is completely determined by its initial value \(\alpha_0\) and, for \(t\in\mathbb{T}\), the collection of portfolios \((\alpha_{t+1}(\mu),\beta_{t+1}(\mu))\) created at all nodes \(\mu\in\Omega_t\). Consequently, we may represent \((\alpha,\beta)\) by a vector in \(\mathrm{R}^{2m+1}\).

In view of Remark~\ref{rem:fundamentals:trading-arbitrage:wedging}, such a strategy \((\alpha,\beta)\in\mathbb{R}^{2m+1}\) is self-financing if and only if for \(t\in\mathbb{T}\backslash\{T\}\) it satisfies the linear system of inequalities
\begin{equation}
\left.\qquad\begin{aligned}
\alpha_tR^c_t-\alpha_{t+1} + \left[\beta_t - \beta_{t+1}\right]S^a_t&\ge0,\\
\alpha_tR^c_t-\alpha_{t+1} + \left[\beta_t - \beta_{t+1}\right]S^b_t&\ge0,\\
\alpha_tR^d_t-\alpha_{t+1} + \left[\beta_t - \beta_{t+1}\right]S^a_t&\ge0,\\
\alpha_tR^d_t-\alpha_{t+1} + \left[\beta_t - \beta_{t+1}\right]S^b_t&\ge0.
\end{aligned}\qquad\right\}\label{eq:fundamentals:trading-arbitrage:Phi-convex-polyhedral}
\end{equation}
It follows that \(\Phi\) is polyhedral when viewed as a subset of \(\mathbb{R}^{2m+1}\).

Applying inequalities~\eqref{eq:fundamentals:trading-arbitrage:Phi-convex-polyhedral}, one may readily verify that \(\Phi\) is non-empty (it contains the zero vector) and that
\begin{align*}
(\alpha+\alpha',\beta+\beta')&\in\Phi\text{ for all }(\alpha,\beta),(\alpha',\beta')\in\Phi,\\
(\lambda\alpha,\lambda\beta)&\in\Phi\text{ for all }(\alpha,\beta)\in\Phi,\lambda\in[0,\infty).
\end{align*}
Thus the set \(\Phi\) is self-financing strategies is a convex cone.
\end{proof}

\begin{varthm}[Lemma \ref{lem:fundamentals:ftap:discount-factors}]
If the model admits no arbitrage, then there exists an \((0,\infty)^2\)-valued \(\mathcal{F}_T\)-measurable random variable \(Z_T\equiv(Z^B_T,Z^S_T)\) such that
\[\mathbb{E}_{\mathbb{Q}}\left(Z_T\cdot \left(\alpha_{T+1},\beta_{T+1}\right)\right)-\alpha_0\le0\]
for all \((\alpha,\beta)\in\Phi\).
\end{varthm}

\begin{proof}
Let
\[\mathcal{A}:=\left\{(-\alpha_0,\alpha_{T+1},\beta_{T+1})|(\alpha,\beta)\in\Phi\right\}.\]
As \(\Phi\) is a convex polyhedral cone (Proposition~\ref{prop:fundamentals:trading-arbitrage:Phi-convex-polyhedral}), it follows that \(\mathcal{A}\) may be represented by a convex polyhedral cone in \(\mathbb{R}^{2n+1}\), where \(n\) is the number of terminal nodes in the model \cite[Theorem 19.3]{rockafellar1997}.
Also define 
\[\mathcal{L}^+:=\conv\left[\left\{\left.\left(0,\mathbf{1}_{\{\omega\}},0\right)\right|\omega\in\Omega\right\}\cup\left\{\left.\left(0,0,\mathbf{1}_{\{\omega\}}\right)\right|\omega\in\Omega\right\}\cup(1,0,0)\right].\]
Here \(\mathbf{1}_{\mathcal{S}}\) denotes the indicator function of a set \(\mathcal{S}\), i.e.\ it is equal to \(1\) on \(\mathcal{S}\), and to zero everywhere else. If one views the representation of the convex set \(\mathcal{L}^+\) in Euclidean space, then its generators correspond to the canonical basis of \(\mathbb{R}^{2n+1}\). It is closed and moreover bounded, since \(\mathcal{L}^+\subset[0,1]^{2n+1}\).

The no-arbitrage condition implies that
\(\mathcal{A}\cap\mathcal{L}^+=\emptyset.\)
Consequently, by the finite dimensional separation theorem \citeaffixed[Theorem~11.1 and Corollary~11.14.2]{rockafellar1997}{cf.}, there exists an \(\mathcal{F}_T\)-measurable random variable \(W_T\equiv\left(W^R_T, W^S_T\right)\) and a number \(W_0\in\mathbb{R}\) such that
\begin{alignat}{2}
\lefteqn{\sup_{(\alpha,\beta)\in\Phi}\left[\mathbb{E}_{\mathbb{Q}}\left(W_T\cdot \left(\alpha_{T+1},\beta_{T+1}\right)\right)-W_0\alpha_0\right]}\nonumber\\
%\begin{alignat}
&\qquad\qquad\qquad\qquad& &=\sup_{(a_0,a_R,a_S)\in\mathcal{A}}\left[\mathbb{E}_{\mathbb{Q}}\left(W_T\cdot \left(a_R,a_S\right)\right)+W_0a_0\right]\nonumber\\
&& &<\inf_{(l_0,l_R,l_S)\in\mathcal{L}^+}\left[\mathbb{E}_{\mathbb{Q}}\left(W_T\cdot \left(l_R,l_S\right)\right)+W_0l_0\right]<\infty.\label{eq:fundamentals:ftap:discount-factors:separation}
%\end{alignat}
\end{alignat}
Since \(0\in\mathcal{A}\), we must for all \((l_0,l_R,l_S)\in\mathcal{L}^+\) have
\begin{equation}\label{eq:fundamentals:ftap:discount-factors:non-negativity}
0 < \mathbb{E}_{\mathbb{Q}}\left(W_T\cdot \left(l_R,l_S\right)\right)+W_0l_0.
\end{equation}
Substituting each of the elements of the generating set of \(\mathcal{L}^+\) for \((l_0,l_R,l_S)\) into~\eqref{eq:fundamentals:ftap:discount-factors:non-negativity}, we deduce that \(W_0\), \(W^R_T\) and \(W^S_T\) are all strictly positive. Since \(\mathcal{A}\) is a cone, it readily follows from~\eqref{eq:fundamentals:ftap:discount-factors:separation} that in fact
\[\sup_{(\alpha,\beta)\in\Phi}\left[\mathbb{E}_{\mathbb{Q}}\left(W_T\cdot \left(\alpha_{T+1},\beta_{T+1}\right)\right)-W_0\alpha_0\right]=0.\]
Dividing by \(W_0>0\) completes the proof.
\end{proof}

\begin{lemma} \label{lem:fundamentals:ftap:conical}
For \(t\in\mathbb{T}\), if an \(\mathcal{F}_t\)-measurable random variable \(Z\equiv\left(Z^B,Z^S\right)\) has the property that
\begin{equation}\label{eq:fundamentals:ftap:conical:1}
\mathbb{E}_{\mathbb{Q}}\left[Z\cdot\left(\alpha,\beta\right)\right]\ge0
\end{equation}
for all \(\mathcal{F}_t\)-measurable random variables \(\alpha\) and \(\beta\) such that
\begin{equation}\label{eq:fundamentals:ftap:conical:2}
\vartheta_t(\alpha,\beta)\ge0,
\end{equation}
then \(Z^R\) and \(Z^S\) are non-negative and
\[Z^RS^b_t\le Z^S\le Z^RS^a_t.\]
\end{lemma}

\begin{proof}
For any \(\nu\in\Omega_t\) we have
\begin{align}
\mathcal{K}(\nu)
&:=\left\{\left.(\gamma,\delta)\in\mathbb{R}^2\right|\vartheta_t(\gamma,\delta)(\nu)\ge0\right\}\nonumber\\
&=\left\{\left.(\gamma,\delta)\in\mathbb{R}^2\right|\gamma+\delta S^b_t(\nu)\ge0, \gamma+\delta S^a_t(\nu)\ge0\right\}\nonumber\\
&=\conv\cone\left\{\left(1,0\right),\left(0,1\right),\left(S^a_t(\nu),-1\right),\left(-S^b_t(\nu),1\right)\right\}.\label{eq:fundamentals:ftap:conical:3}
\end{align}
The set \(\mathcal{K}(\nu)\) is clearly a polyhedral convex cone in \(\mathbb{R}^2\).

The set \(\mathcal{A}\) of all \(\mathcal{F}_t\)-measurable random variables satisfying~\eqref{eq:fundamentals:ftap:conical:2} can now be written as
\[\mathcal{A}=\left\{(\alpha,\beta)\left|(\alpha,\beta)\text{ is } \mathcal{F}_t\text{-measurable and }\left(\alpha(\nu),\beta(\nu)\right)\in\mathcal{K}(\nu)\text{ for }\nu\in\Omega_t\right.\right\}.\]
The set \(\mathcal{A}\) can be represented by a polyhedral convex cone in \(\mathbb{R}^{2n}\), where \(n\) is the number of nodes at time \(t\). The set \(\mathcal{A}\) also has the property that \(\left(\alpha\mathbf{1}_{\nu},\beta\mathbf{1}_{\nu}\right)\in\mathcal{A}\) for all \(\nu\in\Omega_t\), whenever \((\alpha,\beta)\in\mathcal{A}\). Inequality~\eqref{eq:fundamentals:ftap:conical:1} therefore implies that
\[\mathbb{Q}(\nu)\left[Z(\nu)\cdot(\gamma,\delta)\right]\ge0,\]
i.e.
\[Z(\nu)\cdot(\gamma,\delta)\ge0\]
for all \(\nu\in\Omega_t\) and \((\gamma,\delta)\in\mathcal{K}(\nu)\). By virtue of the representation~\eqref{eq:fundamentals:ftap:conical:3}, it follows that $Z^R(\nu)\ge0$, $Z^S(\nu)\ge0$ and
\begin{align*}
S^a_t(\nu)Z^R(\nu) - Z^S(\nu) &\ge0, &
-S^b_t(\nu)Z^R(\nu) + Z^S(\nu) &\ge0
\end{align*}
for \(\nu\in\Omega_t\).
\end{proof}

\begin{lemma} \label{lem:fundamentals:ftap:kkt}
Suppose that \(\mathcal{A}\subseteq\mathbb{R}^n\) and \(\mathcal{C}\subseteq\mathbb{R}^k\) are polyhedral convex sets, that \(\mathcal{C}\) is a cone, and that \(L:\mathbb{R}^n\rightarrow\mathbb{R}\) and \(G:\mathbb{R}^n\rightarrow\mathbb{R}^k\) are affine functions. If there exists an element \(z\) of \(\mathcal{A}\) that satisfies \(G(z)\in\mathcal{C}\) and \(L(z)\ge L(x)\) for all \(x\in\mathcal{A}\) satisfying
\(G(x)\in\mathcal{C},\)
then there exists a linear functional \(W:\mathbb{R}^k \rightarrow\mathbb{R}\) such that \(W(c)\ge0\) for all \(c\in\mathcal{C}\),
\[L(z)+W(G(z)) \ge L(x) + W(G(x)) \text{ for all }x\in\mathcal{A}\]
and
\[W(G(z)) = 0.\]
\end{lemma}

\begin{proof}
\citeasnoun[Theorem~8.3.1]{luenberger1969} and \citeasnoun[Corollary 28.3.1]{rockafellar1997}, among others.
\end{proof}

\begin{varthm}[Lemma \ref{lem:fundamentals:ftap:pricing-system}]
Suppose that the model admits a consistent discount factor. There exists an \((0,\infty)^2\)-valued adapted process \(Z\equiv(Z^B,Z^S)\) such that
\begin{equation}\tag{\ref{eq:fundamentals:ftap:pricing-system:wedge-1}}
S^b_t\le\frac{Z^S_t}{Z^B_t}\le S^a_t\text{ for }t\in\mathbb{T}
\end{equation}
and \(Z^S\) is a martingale under \(\mathbb{Q}\), and a predictable process \(R\) satisfying
\begin{equation}\tag{\ref{eq:fundamentals:ftap:pricing-system:wedge-2}}
R^d_t\le R_t\le R^c_t \text{ for } t\in\mathbb{T}
\end{equation}
such that \(BZ^R\) is a martingale under \(\mathbb{Q}\), where $B$ is the deflator process associated with $R$.
\end{varthm}

\begin{proof}
We adapt and extend the argument of \citeasnoun[Theorem~7.1]{dempster_evstigneev_taksar2006} directly. We say that a quadruple \(\left(\alpha,\beta,\gamma,\delta\right)\) is a \emph{generalised trading strategy} if \((\alpha,\beta)\) is a trading strategy, the real-valued processes \(\gamma\) and \(\delta\) are adapted and for all \(t\in\mathbb{T}\) we have
\begin{align*}
\varrho(\alpha_t)&\ge\gamma_t,&
\beta_t&\ge\delta_t, &\vartheta_t\left(\gamma_t-\alpha_{t+1},\delta_t-\beta_{t+1}\right)\ge0.\end{align*}
Denote the collection of generalised trading strategies by \(\Theta\).

A generalised trading strategy is simply a self-financing trading strategy with explicit consumption. At any trading date \(t\in\mathbb{T}\), the owner of the generalised trading strategy \(\left(\alpha,\beta,\gamma,\delta\right)\in\Theta\) starts trading with his/her current portfolio of \(\varrho(\alpha_t)\) in cash and \(\beta_t\) stocks (recall that \(\beta_0=0\)). He/she then consumes \(\varrho(\alpha_t)-\gamma_t\ge0\) in cash and \(\beta_t-\delta_t\ge0\) units of stock. The resulting portfolio of \(\gamma_t\) in cash and \(\delta_t\) stocks is then used to change in a self-financing way into a portfolio of \(\alpha_{t+1}\) in cash and \(\beta_{t+1}\) stocks, which is held until the next trading date. 

Suppose that the model admits a consistent discount factor \(Z_T\equiv\left(Z^S_T,Z^R_T\right)\), and define
\begin{multline*}
\mathcal{A}:=\left\{\left.\left(\alpha,\beta,\gamma,\delta\right)\right| \left(\alpha,\beta\right) \text{ predictable, }\beta_0=0, \left(\gamma,\delta\right) \text{ adapted, }\right.\\
\left.\vartheta_t\left(\gamma_t-\alpha_{t+1},\delta_t-\beta_{t+1}\right)\ge0\text{ for }t\in\mathbb{T}\right\}.
\end{multline*}
The set \(\mathcal{A}\) is clearly a polyhedral convex subset of \(\mathbb{R}^{4m+1}\), where \(m\) is the number of nodes in the model (recall our convention of adding final portfolios to trading strategies). On this set, define the linear functional \(L\) by
\[L\left(\alpha,\beta,\gamma,\delta\right):=\mathbb{E}_{\mathbb{Q}}\left(Z_T\cdot \left(\alpha_{T+1},\beta_{T+1}\right)\right)-\alpha_0.\]

Consider now the problem of maximising \(L\) over all \(\left(\alpha,\beta,\gamma,\delta\right)\in\mathcal{A}\), subject to the constraints
\begin{align}
R^d_t\alpha_t-\gamma_t &\ge 0, & R^c_t\alpha_t-\gamma_t &\ge 0, & \beta_t-\delta_t &\ge 0 \label{eq:fundamentals:ftap:pricing-system:slack:1}
\end{align}
for all \(t\in\mathbb{T}\). Observe that inequalities~\eqref{eq:fundamentals:ftap:pricing-system:slack:1} are equivalent to
\begin{align*}
\varrho_t(\alpha_t)-\gamma_t&\ge0, & \beta_t-\delta_t &\ge 0,
\end{align*}
and therefore our optimisation problem is none other than the problem of maximising \(L\) on the set \(\Theta\) of generalised trading strategies.

Fix a generalised trading strategy \(\left(\alpha,\beta,\gamma,\delta\right)\in\Theta\). If we now construct the self-financing strategy \((\varepsilon,\eta)\in\Phi\) by setting \(\varepsilon_0:=\alpha_0\) and
\begin{align*}
\eta_t&:=\beta_t \text{ for }t\in\mathbb{T}\cup\{T+1\},\\
\varepsilon_{t+1}&:=\vartheta_t\left(\varrho_t(\varepsilon_t),\eta_{t+1}-\eta_t\right)\text{ for }t\in\mathbb{T},
\end{align*}
then it readily follows that \(\varepsilon_t\ge\alpha_t\) and \(\eta_t\ge\beta_t\) for all \(t\in\mathbb{T}\cup\{T+1\}\). Consequently, Lemma~\ref{lem:fundamentals:ftap:discount-factors} yields
\[\sup_{\left(\alpha,\beta,\gamma,\delta\right)\in\Theta}L\left(\alpha,\beta,\gamma,\delta\right)\le\sup_{\left(\varepsilon,\eta\right)\in\Phi}\left[\mathbb{E}_{\mathbb{Q}}\left(Z_T\cdot \left(\varepsilon_{T+1},\eta_{T+1}\right)\right)-\varepsilon_0\right]\le0.\]
Moreover, this supremum is attained at \(0\in\Theta\subset\mathcal{A}\).

Application of Lemma~\ref{lem:fundamentals:ftap:kkt} guarantees the existence of an adapted \([0,\infty)^3\)-valued process \(W\equiv\left(W^d,W^c,W^S\right)\) such that
\begin{equation} \label{eq:fundamentals:ftap:pricing-system:kkt:short}
L\left(\alpha,\beta,\gamma,\delta\right)+\sum_{t\in\mathbb{T}}\mathbb{E}_{\mathbb{Q}}\left[W_t\cdot\left(R^d_t\alpha_t-\gamma_t, R^c_t\alpha_t-\gamma_t, \beta_t-\delta_t\right)\right]\le0
\end{equation}
for all \(\left(\alpha,\beta,\gamma,\delta\right)\in\mathcal{A}\). We may rearrange~\eqref{eq:fundamentals:ftap:pricing-system:kkt:short} as
\begin{multline}
W_0\cdot\left(R^d_0\alpha_0, R^c_0\alpha_0, \beta_0\right) - \alpha_0\\
+ \sum_{t\in\mathbb{T}\backslash\{T\}}\mathbb{E}_{\mathbb{Q}}\left[W_{t+1}\cdot\left(R^d_{t+1}\alpha_{t+1}, R^c_{t+1}\alpha_{t+1}, \beta_{t+1}\right) - W_t\cdot\left(\gamma_t, \gamma_t, \delta_t\right)\right]\\
+ \mathbb{E}_{\mathbb{Q}}\left[Z_T\cdot\left(\alpha_{T+1}, \beta_{T+1}\right) - W_T\cdot\left(\gamma_T, \gamma_T, \delta_T\right)\right] \le 0. \label{eq:fundamentals:ftap:pricing-system:kkt:expanded}
\end{multline}

Recall from the definition of \(\mathcal{A}\) that we may substitute \(0\) for any element \(\left(\alpha_{t+1},\beta_{t+1},\gamma_t,\delta_t\right)\) of \(\left(\alpha,\beta,\gamma,\delta\right)\in\mathcal{A}\) at any time \(t\in\mathbb{T}\). With this substitution in hand, we may now deduce a number of inequalities from~\eqref{eq:fundamentals:ftap:pricing-system:kkt:expanded}. First, for all \(\alpha_0\in\mathbb{R}\) we have
\[W_0\cdot\left(R^d_0\alpha_0, R^c_0\alpha_0, \beta_0\right) - \alpha_0 \le 0,\]
i.e.
\begin{equation}
\left(W^d_0+W^c_0-1\right)\alpha_0 \le 0.\label{eq:fundamentals:ftap:pricing-system:kkt:1}
\end{equation}
Moreover, for each \(t\in\mathbb{T}\backslash\{T\}\) and for all \(\mathcal{F}_t\)-measurable random variables \(\alpha_{t+1}\), \(\beta_{t+1}\), \(\gamma_t\) and \(\delta_t\) such that
\[\vartheta_t\left(\gamma_t-\alpha_{t+1},\delta_t-\beta_{t+1}\right)\ge0,\]
we have
\begin{equation}
\mathbb{E}_{\mathbb{Q}}\left[W_{t+1}\cdot\left(R^d_{t+1}\alpha_{t+1}, R^c_{t+1}\alpha_{t+1}, \beta_{t+1}\right) - W_t\cdot\left(\gamma_t, \gamma_t, \delta_t\right)\right] \le 0.\label{eq:fundamentals:ftap:pricing-system:kkt:2}
\end{equation}
Finally, for all \(\mathcal{F}_T\)-measurable random variables \(\alpha_{T+1}\), \(\beta_{T+1}\), \(\gamma_T\) and \(\delta_T\) such that
\[\vartheta_T\left(\gamma_T-\alpha_{T+1},\delta_T-\beta_{T+1}\right)\ge0,\]
we have
\begin{equation}
\mathbb{E}_{\mathbb{Q}}\left[Z_T\cdot\left(\alpha_{T+1}, \beta_{T+1}\right) - W_T\cdot\left(\gamma_T, \gamma_T, \delta_T\right)\right] \le 0.\label{eq:fundamentals:ftap:pricing-system:kkt:3}
\end{equation}

Inequality~\eqref{eq:fundamentals:ftap:pricing-system:kkt:1} clearly implies that
\[W^d_0 + W^c_0=1>0.\]
Fix now any \(t\in\mathbb{T}\backslash\{T\}\), and any \(\mathcal{F}_t\)-measurable random variable \(\lambda\). Choosing
$\alpha_{t+1}=\gamma_t=\lambda$ and $\beta_{t+1}=\delta_t=0$ in inequality~\eqref{eq:fundamentals:ftap:pricing-system:kkt:2} yields
\[\mathbb{E}_{\mathbb{Q}}\left(\left[R^d_{t+1}W^d_{t+1} + R^c_{t+1}W^c_{t+1} - W^d_t - W^c_t\right]\lambda\right) \le 0.\]
We readily deduce that
\begin{align}
W^d_t + W^c_t
&= \mathbb{E}_{\mathbb{Q}}\left(\left.R^d_{t+1}W^d_{t+1} + R^c_{t+1}W^c_{t+1}\right|\mathcal{F}_t\right) \nonumber\\
&= R^d_{t+1}\mathbb{E}_{\mathbb{Q}}\left(\left.W^d_{t+1}\right|\mathcal{F}_t\right) + R^c_{t+1}\mathbb{E}_{\mathbb{Q}}\left(\left.W^c_{t+1}\right|\mathcal{F}_t\right).\label{eq:fundamentals:ftap:pricing-system:kkt:4}
\end{align}
By analogy, we also obtain
\begin{equation}
\mathbb{E}_{\mathbb{Q}}\left(\left.W^S_{t+1}\right|\mathcal{F}_t\right)=W^S_t.\label{eq:fundamentals:ftap:pricing-system:kkt:5}
\end{equation}
Finally, at time \(T\), if we choose $\alpha_{T+1}=\gamma_T=\lambda$ and $\beta_{T+1}=\delta_T=0$ for any random variable \(\lambda\), then inequality~\eqref{eq:fundamentals:ftap:pricing-system:kkt:4} becomes
\[\mathbb{E}_{\mathbb{Q}}\left(\left[Z^R - W^d_T - W^c_T\right]\lambda\right) \le 0.\]
It follows from the non-degeneracy of \(\mathbb{Q}\) that
\begin{equation}\label{eq:fundamentals:ftap:pricing-system:final-inequality:1}
W^d_T + W^c_T = Z^R_T > 0.
\end{equation}
In a similar fashion, we deduce that
\begin{equation}\label{eq:fundamentals:ftap:pricing-system:final-inequality:2}
W^S_T = Z^S_T > 0.
\end{equation}

We are now in a position to directly define the process \(Z\equiv\left(Z^R,Z^S\right)\) with final value \(Z_T\). Indeed, for \(t\in\mathbb{T}\backslash\{T\}\) set
\begin{align*}
Z^R_t &:= W^d_t + W^c_t, &
Z^S_t &:= W^S_t.
\end{align*}
Equations~\eqref{eq:fundamentals:ftap:pricing-system:kkt:5} and~\eqref{eq:fundamentals:ftap:pricing-system:final-inequality:2} permit us to conclude that \(Z^S\) is a strictly positive martingale under \(\mathbb{Q}\). It follows from~\eqref{eq:fundamentals:ftap:pricing-system:kkt:4} that there exists a predictable process \(R\) such that 
\[R^d_{t+1}\le R_{t+1}\le R^c_{t+1}\]
and
\[Z^R_t=R_{t+1}\mathbb{E}_{\mathbb{Q}}\left(\left.Z^R_{t+1}\right|\mathcal{F}_t\right)\]
for \(t\in\mathbb{T}\backslash\{T\}.\) Also let \(R_0:=1\). We also conclude from~\eqref{eq:fundamentals:ftap:pricing-system:final-inequality:1} by induction that \(Z^R\) is strictly positive.
Finally, upon taking \(\alpha_{t+1}=0\) and \(\beta_{t+1}=0\) for \(t\in\mathbb{T}\backslash\{T\}\) in~\eqref{eq:fundamentals:ftap:pricing-system:kkt:2} and for \(t=T\) in~\eqref{eq:fundamentals:ftap:pricing-system:kkt:3}, Lemma~\ref{lem:fundamentals:ftap:conical} permits us to conclude that 
\[S^b_t\le\frac{Z^S_t}{Z^B_t}\le S^a_t \text{ for }t\in\mathbb{T}.\qedhere\]
\end{proof}

\bibliography{marketswithfrictions.bib}

\end{document}